\documentclass{article}
\usepackage{authblk}

\usepackage[T1]{fontenc}
\usepackage[utf8x]{inputenc}
\usepackage[english]{babel}

\usepackage{mathtools,amssymb,amsthm}

\usepackage{enumitem,multirow}

\renewcommand{\ge}{\geqslant}
\renewcommand{\leq}{\leqslant}
\renewcommand{\geq}{\geqslant}
\renewcommand{\phi}{\varphi}

\newcommand{\R}{\mathbb{R}}

\newcommand{\W}{\mathcal{W}}
\newcommand{\Lo}{\mathcal{L}}

\newcommand{\0}{\boldsymbol{0}}
\newcommand{\1}{\boldsymbol{1}}

\newcommand{\ip}[2]{\langle{#1},{#2}\rangle}

\newcommand{\conv}{\mathrm{conv}}
\newtheorem{thm}{Theorem}

\newtheorem{conj}[thm]{Conjecture}
\theoremstyle{definition}

\newtheorem{rem}[thm]{Remark}
\newtheorem{ex}[thm]{Example}

\makeatletter
\newtheorem*{rep@theorem}{\rep@title}
\newcommand{\newreptheorem}[2]{%
\newenvironment{rep#1}[1]{%
 \def\rep@title{#2 \ref{##1}}%
 \begin{rep@theorem}}%
 {\end{rep@theorem}}}
\makeatother

\newreptheorem{thm}{Theorem}

\begin{document}

\title{On Critical Threshold Value for Simple Games}

\author{Kanstantsin Pashkovich}

\affil{Department of Combinatorics and Optimization,\\ University of Waterloo,\\ 200 University Avenue West
Waterloo, ON, Canada N2L 3G1\\
  \texttt{kpashkov@uwaterloo.ca}}

\date{}

\maketitle

\begin{abstract}
In this note, we show  that for every simple game with $n$ players the critical threshold value is at most $n/4$.  This verifies the conjecture of Freixas and Kurz. 
\end{abstract}

\section{Introduction}

Let $N$ be a finite set of players. We call a function $v:2^N\rightarrow \{0,1\}$ \emph{monotone} if $v(C)\leq v(S)$ for all $C,S \subseteq N$ such that $C\subseteq S$. The pair $(N,v)$ is called a \emph{simple game} if $v(\varnothing)=0$, $v(N)=1$ and $v$ is a monotone $0/1$ function. We refer the reader to~\cite{vNM44},~\cite{Shapley62} for further reading on simple games. A simple game $(N,v)$ partitions the collection of all possible player coalitions into two collections: the \emph{collection of winning coalitions} $\W:=\{C\subseteq N\,:\, v(C)=1\}$ and the \emph{collection of losing coalitions} $\Lo:=\{C\subseteq N\,:\, v(C)=0\}$.

\emph{Weighted voting games} are a natural family of simple games. A weighted voting game is defined by a finite set of players $N$ and a vector $p\in \R^N$, $p\geq \0$, $p(N)\geq 1$, where
\[
v(C):=
\begin{cases}
1 &\text{if}\quad p(C)\geq 1\\
0 &\text{otherwise}\,.
\end{cases}
\]
Here, we use the notation $q(C):=\sum_{i\in C}q_i$ for a vector $q\in \R^N$ and $C\subseteq N$.

Clearly, every weighted voting game is a simple game. To show that the reverse is not true let us consider the following example from~\cite{FK14}. 
\begin{ex}\label{example}
Let $N=\{1,2,\ldots,n\}$ for some even $n$ and the value function $v:2^N\rightarrow \{0,1\}$ be as follows
\[
v(C):=
\begin{cases}
1 &\text{if}\quad\{2i-1,2i \}\subseteq C \quad \text{for some}\quad i\in \{1,\ldots,n/2\}\\
0 &\text{otherwise}\,.
\end{cases}
\]
Obviously, coalitions $\{2i-1, 2i\}$, $i\in \{1,\ldots,n/2\}$ are winning while the two coalitions $\{1,3,\ldots, n-1\}$, $\{2,4,\ldots, n\}$ are losing. 
\end{ex}
If the desired vector $p\in \R^N$ exists for Example~\ref{example}, then on one side $p(N)\geq n/2$ and on the other side $p(N)<2$,  showing that for $n\geq 4$ this game is not a weighted voting game.

To understand whether a simple game is a weighted voting game, we could use the \emph{critical threshold value} introduced in~\cite{GHS13}. Before we define the critical threshold value of a simple game, let us define the following polyhedron
\[
Q(\W):=\{x\in\R^N\,:\, x(C) \ge \1 \text{  for  } C\in\W\,,\quad x\ge \0\}\,.
\]
The \emph{critical threshold value} can be defined as
\[
	\alpha=\alpha(N,v):=\min_{p\in Q(\W)} \max_{C\in \Lo} p(C)\,.
\]
Observe, that $\alpha<1$ if and only if the simple game $(N,v)$ is a weighted voting game.
 
The example in~\cite{FK14} shows that $\alpha$ can be as large as $n/4$, because $\frac{2}{n}\1$ lies in the convex hull of the characteristic vectors of winning coalitions while $\frac{1}{2}\1$ lies in the convex hull of the characteristic vectors of losing coalitions. Freixas and Kurz~\cite{FK14}  conjectured that there is no simple game with a larger value of $\alpha$. Here, we state the variant of the conjecture of Freixas and Kurz from~\cite{HKKP18}.

\begin{conj}[Conjecture of Freixas and Kurz]\label{conj}
For a simple game with $n$ players, the collection of winning coalitions $\W$ and the collection of losing coalitions $\Lo$, we have
\[
\alpha=\min_{p\in Q(\W)} \max_{C\in \Lo} p(C)\leq n/4\,.
\]
\end{conj}

In~\cite{HKKP18} the conjecture of Freixas and Kurz was verified for simple games with all minimal wining coalitions of size $3$ and for simple games with no minimal winning coalitions of size $3$. In~\cite{HKKP18} it was shown that $\alpha\leq 2n/7$ for general simple games.

Before going to the proof, we would like to say that our approach is inspired by the work of Ahmad Abdi, G\'{e}rard Cornu\'{e}jols and Dabeen Lee on identically self-blocking clutters~\cite{A18} (Section 3).

\section{Proof}

To prove the conjecture we reformulate, strengthen and only then verify it. A coalition~$C$, $C\subseteq N$ is called a \emph{cover} of $\W$ if $C$ has at least one common player with every coalition in $\W$. We call the collection of covers of $\W$ the \emph{blocker} of $\W$ and denote it by $b(\W)$\footnote{Usually, blocker is defined as the collection of minimal covers. Here, for simplicity of exposition we define blocker as the collection of all covers.} \cite{Edmonds70}. Due to the definition of simple games, we have
\[
\Lo=\{ N\setminus C\,:\, C\in b(\W)\}\,.
\]
Hence, the critical threshold value can be reformulated as follows
\begin{align*}
\alpha=\min_{p\in Q(\W)} \max_{L\in \Lo} p(L)=
 &\min_{p\in Q(\W)} \max_{C\in b(\W)} p(N\setminus C)=\\
&\min_{p\in Q(\W)} \max_{\substack{q\in Q(\W)\\q\in \{0,1\}^N}} \ip{p}{\1-q}\,.
\end{align*}
Here, $\ip{p}{q}$ stands for the scalar product of two vectors $p$ and $q$.

\begin{conj}[Reformulation of Conjecture of Freixas and Kurz]\label{conj_reform}
For a simple game with $n$ players and the collection of winning coalitions $\W$, we have
\[
\min_{p\in Q(\W)} \max_{\substack{q\in Q(\W)\\q\in \{0,1\}^N}} \ip{p}{\1-q}\leq n/4\,.
\]
\end{conj}

Next, we prove Theorem~\ref{thm_strength}, which is a strengthening of Conjecture~\ref{conj}. For the proof we need the following straightforward remark, which we leave as an exercise.

\begin{rem}\label{rem_min_norm}
Let $P$ be a polyhedron and let $p^\star$ be the optimal solution of the program $\min \{\left\lVert p \right\rVert_2\,:\, p\in P\}$. Then $p^\star$ is an optimal solution of the linear program $\min \{\ip{p^\star}{q}\,:\, q\in P\}$.
\end{rem}

\begin{thm}[Strengthening of Conjecture of Freixas and Kurz]\label{thm_strength}
For a simple game with $n$ players and the collection of winning coalitions $\W$, we have
\[
\min_{p\in Q(\W)} \max_{q\in Q(\W)} \ip{p}{\1-q}\leq n/4\,.
\]
In particular, if $p^\star$ is the optimal solution for the program
\[
\min \{\left\lVert p \right\rVert_2\,:\, p\in Q(\W)\}\,,
\]
then
\[
 \max_{q\in Q(\W)} \ip{p^\star}{\1-q}\leq n/4\,.
\]
\end{thm}
\begin{proof}

Let us consider the unique optimal solution $p^\star$ for the program $\min \{\left\lVert p \right\rVert_2\,:\, p\in Q(\W)\}$. By Remark~\ref{rem_min_norm}, $p^\star$ is an optimal solution for the program $\min \{\ip{p^\star}{q}\,:\, q\in Q(\W)\}$. Thus, $p^\star$ is an optimal solution for the program $\max_{q\in Q(\W)} \ip{p^\star}{\1-q}$. Thus, we have
\[
\max_{q\in Q(\W)} \ip{p^\star}{\1-q}=\ip{p^\star}{\1-p^\star}=\frac{n}{4}-\ip{\frac{1}{2}\1-p^\star}{\frac{1}{2}\1-p^\star}\leq \frac{n}{4}\,,
\]
finishing the proof.
\end{proof}

To finish the note, let us discuss when Conjecture~\ref{conj} provides a tight upper bound for the critical threshold value. The next theorem shows that if the upper bound in Conjecture~\ref{conj} is tight, then this fact can be certified in the same way as in Example~\ref{example}.

\begin{thm}
For a simple game with $n$ players and the collection of winning coalitions $\W$ and the collection of losing coalitions $\Lo$, we have
\[
\alpha=\min_{p\in Q(\W)} \max_{L\in \Lo} p(L)= n/4
\]
if and only if $\frac{2}{n}\1$ lies in the convex hull of the characteristic vectors of winning coalitions and $\frac{1}{2}\1$ lies in the convex hull of the characteristic vectors of losing coalitions.
\end{thm}
\begin{proof}
Clearly, if $\frac{2}{n}\1$ lies in the convex hull of the characteristic vectors of winning coalitions and $\frac{1}{2}\1$ lies in the convex hull of the characteristic vectors of losing coalitions, then for every $p\in Q(\W)$ we have
\[
\max_{L\in \Lo} p(L)\geq \ip{p}{\frac{1}{2}\1}= \frac{n}{4}\ip{p}{\frac{2}{n}\1}\geq \frac{n}{4}\,,
\]
showing that $\alpha\geq n/4$ and hence $\alpha= n/4$ by Theorem~\ref{thm_strength}.

On the other hand, from the proof of Theorem~\ref{thm_strength} we know that if $\alpha=n/4$ then $p^\star=\frac{1}{2}\1$ is an optimal solution for $\min \{\ip{p^\star}{q}\,:\, q\in Q(\W)\}$ with value $n/4$. Let us show that $\frac{2}{n}\1$ lies in the convex hull of the characteristic vectors of winning coalitions. To do that consider an optimal dual  solution $y^\star$ for the program $\min \{\ip{p^\star}{q}\,:\, q\in Q(\W)\}$. Using complementary slackness it is straightforward to show that $\frac{4}{n}y^\star$ provides coefficients of a convex combination of characteristic vectors of winning coalitions, where the convex  combination equals~$\frac{2}{n}\1$.

In the same way as the proof of Theorem~\ref{thm_strength}, we could show that
\[
\alpha\leq\max_{\substack{q\in Q(\W)\\q\in \{0,1\}^N}} \ip{q^\star}{\1-q}=\ip{q^\star}{\1-q^\star}=\frac{n}{4}-\ip{\frac{1}{2}\1-q^\star}{\frac{1}{2}\1-q^\star}\leq \frac{n}{4}\,,
\]
where $q^\star$ is the optimal solution for the program 
\[
\min \{\left\lVert q \right\rVert_2\,:\, q\in \conv\{r\in \{0,1\}^N\,:\,r\in Q(\W)\}\}\,.
\]
Thus, if $\alpha$ equals $n/4$, then $q^\star=\frac{1}{2}\1$ and $\frac{1}{2}\1$ lies in $\conv\{r\in \{0,1\}^N\,:\,r\in Q(\W)\}$. Hence, if $\alpha$ equals $n/4$, then $\1-q^\star=\frac{1}{2}\1$ lies in the convex hull of the characteristic vectors of losing coalitions, finishing the proof.
\end{proof}

\section{Open Questions}
The question about asymptotic behaviour of the critical threshold value of complete simple games remains open. These are the games with a total order of players by "winning power". Freixas and Kurz~\cite{FK14} conjectured that the critical threshold value of a complete simple game with $n$ players equals~$O(\sqrt n)$. Recently, in~\cite{HKKP18} it was shown that the critical threshold value of such games is $O\left((\ln n) \sqrt n\right)$.

\subsection*{Acknowledgements.} We would like to thank Ahmad Abdi  for helpful comments on the first version of this note. 

\bibliographystyle{amsplain}
\bibliography{bibliography}

\end{document}